\documentclass[letterpaper,11pt]{article}

\usepackage{times}

\usepackage{amsthm,amsmath,amsfonts}

\usepackage{graphicx}
\usepackage[small]{caption}
\usepackage[hidelinks]{hyperref}
\newcommand{\orcid}[1]{\,\href{https://orcid.org/#1}{\includegraphics[width=8pt]{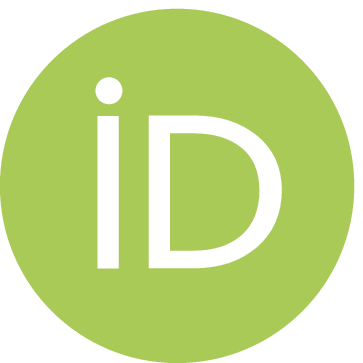}}}

\usepackage{color}

\renewcommand{\S}{\mathcal{S}}

\newcommand{\scirc}{\textsc{Simple Circuit}}
\newcommand{\spath}{\textsc{Simple Path}}

\newtheorem{theorem}{Theorem}
\newtheorem{lemma}{Lemma}
\theoremstyle{remark}

\title{Linking disjoint segments into a simple polygon is hard}

\author{Rain Jiang\orcid{0000-0002-0144-942X}\qquad
Kai Jiang\orcid{0000-0001-8165-0571}\qquad
Minghui Jiang\orcid{0000-0003-1843-9292}\,\thanks{\texttt{ dr.minghui.jiang at gmail.com}}\medskip\\
Home School, USA}

\date{}

\begin{document}

\maketitle

\begin{abstract}
	Deciding whether a family of disjoint line segments in the plane
	can be linked into a simple polygon (or a simple polygonal chain)
	by adding segments between their endpoints is NP-hard.
\end{abstract}

\section{Introduction}

Given a family $\S$ of $n$ closed line segments in the plane,
\scirc\ (respectively, \spath)
is the problem of deciding whether
these segments can be linked into
a simple polygon (respectively, a simple polygonal chain)
by adding segments between their endpoints.

Rappaport~\cite{Ra89} proved that \scirc\ is NP-hard
if the segments in $\S$ are allowed to intersect at their common endpoints,
and asked whether the problem remains NP-hard when the segments are disjoint.
Later, Bose, Houle, and Toussaint~\cite{BHT01} asked
whether the related problem \spath\ is NP-hard
when the segments in $\S$ are disjoint.
More recently, T\'oth~\cite{To06} asked about the complexity of \spath\ again,
with a special interest in the case
when the segments in $\S$ are both disjoint and axis-parallel.

In this note,
we prove the following theorem:

\begin{theorem}\label{thm:hard}
	\scirc\ and \spath\ are both NP-hard
	even if the segments in $\S$ are disjoint and have only four distinct orientations.
\end{theorem}

We prove the theorem in two steps.
First, we modify the construction in Rappaport's proof
of NP-hardness of \scirc\ on not necessarily disjoint segments~\cite{Ra89}
to show that the problem remains NP-hard on disjoint segments.
Next, we modify the construction further
to prove that \spath\ is also NP-hard on disjoint segments.
In the following, we briefly review Rappaport's proof,
which is based on a polynomial reduction from the NP-hard problem
\textsc{Hamiltonian Path} in planar cubic graphs.

\medskip
For any family $\S$ of closed segments in the plane,
denote by $V(\S)$ the set of endpoints of the segments in $\S$.
For any two endpoints $p$ and $q$ in $V(\S)$,
we call the open segment $pq$ a \emph{visibility edge}
if it does not intersect any closed segment in $\S$.

\medskip
Given a planar cubic graph $G$ with $n \ge 4$ vertices,
the reduction~\cite{Ra89} first obtains a rectilinear planar layout of the planar graph
using an algorithm of Rosenstiehl and Tarjan~\cite{RT86},
then constructs a family $\S$ of $O(n)$ segments
following the rectilinear planar layout,
such that $G$ admits a Hamiltonian path
if and only if $\S$ can be linked into a simple polygon
by adding visibility edges between points in $V(\S)$.
The segments in $\S$ are axis-parallel and interior-disjoint,
but may intersect at common endpoints.
Each endpoint in $V(\S)$ is incident to at most one horizontal segment
and at most one vertical segment in $\S$.
Since the rectilinear planar layout
is of width at most $2n-4$ and of height at most $n$~\cite{RT86},
all coordinates of endpoints in $V(\S)$ are integers of magnitude
$O(n^2)$;
indeed a closer look at the construction~\cite[Figure~8]{Ra89}
shows that
$V(\S) \subseteq [1, 22n^2]\times [1, 11n]$.
The reduction is hence strongly polynomial.
Consequently, \scirc\ is strongly NP-hard,
on not necessarily disjoint segments.

\section{Modification for \scirc}

To show that \scirc\ remains strongly NP-hard on disjoint segments,
we will transform the family $\S$ of interior-disjoint axis-parallel segments,
which Rappaport constructed,
into a family $\S'$ of disjoint segments with four distinct orientations,
in polynomial time,
such that $\S$ can be linked into a simple polygon
if and only if $\S'$ can be linked into a simple polygon.
This transformation can be viewed as a reduction from \scirc\ on one type
of input to the same problem on another type of input.

\begin{figure}[htb]
	\centering\includegraphics{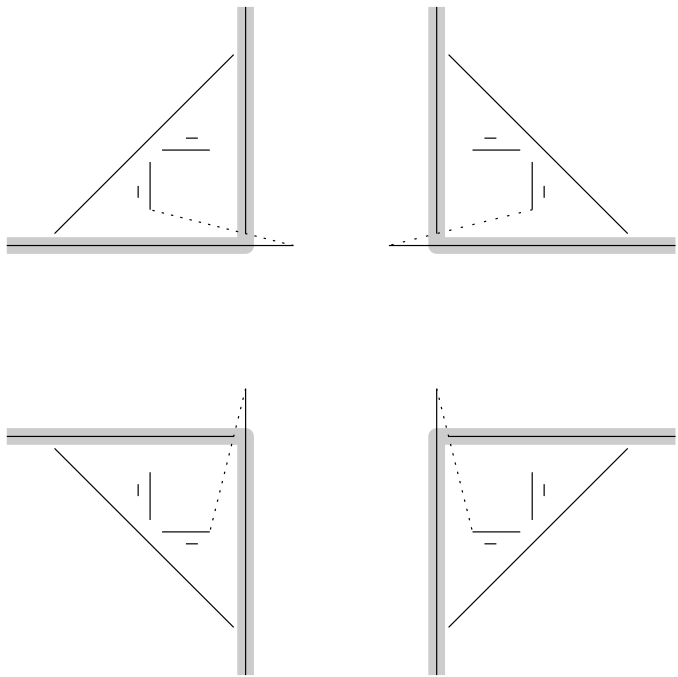}
	\caption{Gadgets for intersections of four different orientations.}\label{fig:turn4}
\end{figure}

To obtain $\S'$ from $\S$,
we first scale the integer coordinates of all segment endpoints
by a factor of $40$,
then locally modify each intersection between a horizontal segment
and a vertical segment into a gadget.
The gadgets come in four variants, one for each possible orientation of an intersection;
see Figure~\ref{fig:turn4}.
For a global picture of these modifications on a rectilinear simple polygon
with intersections of all four different orientations,
refer to Figure~\ref{fig:circuit}.

\begin{figure}[htbp]
	\centering\includegraphics{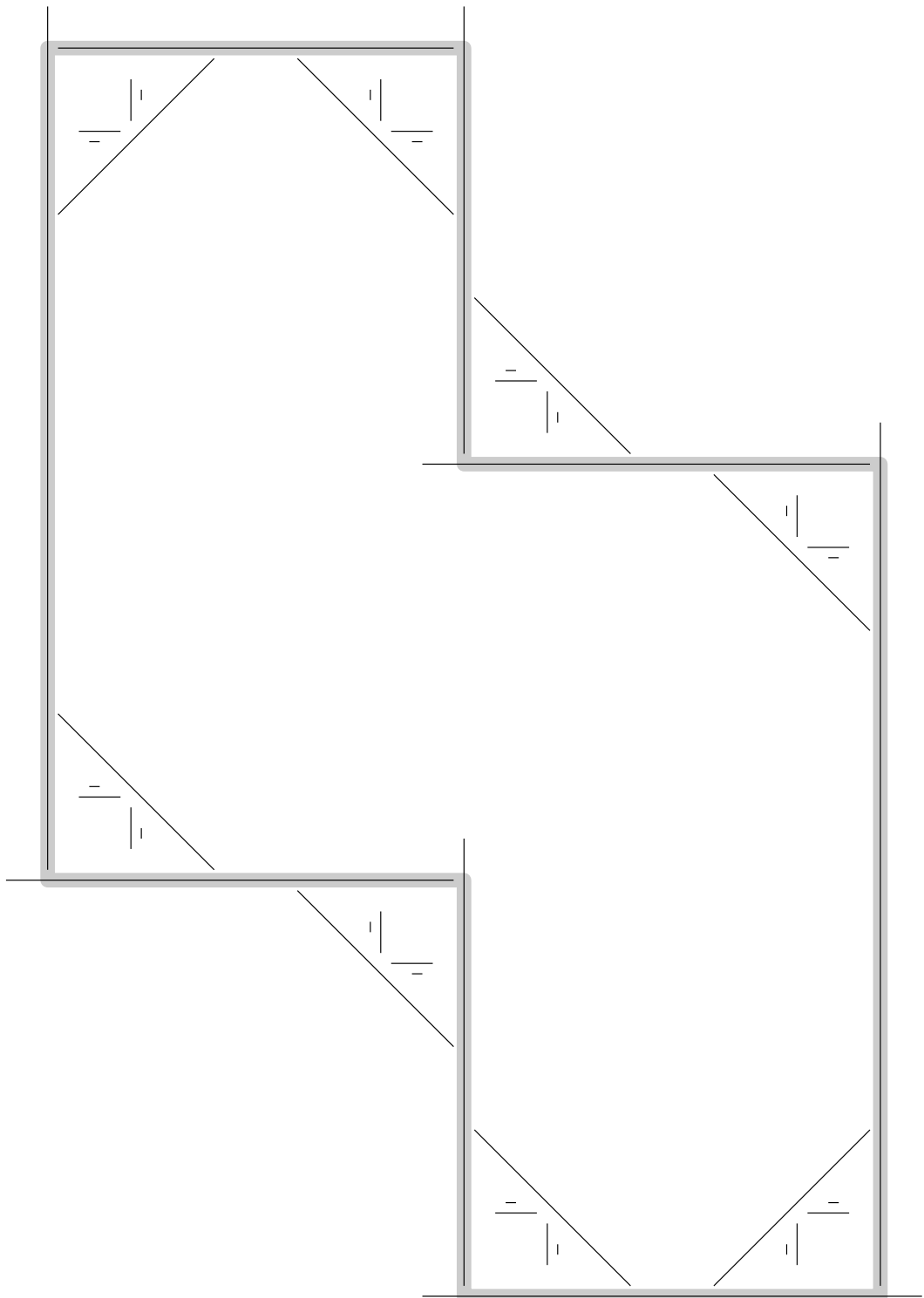}
	\caption{Local modifications of the intersections a rectilinear simple polygon
	into gadgets of disjoint segments.}\label{fig:circuit}
\end{figure}
\clearpage

\begin{figure}[htb]
	\includegraphics[width=.48\textwidth]{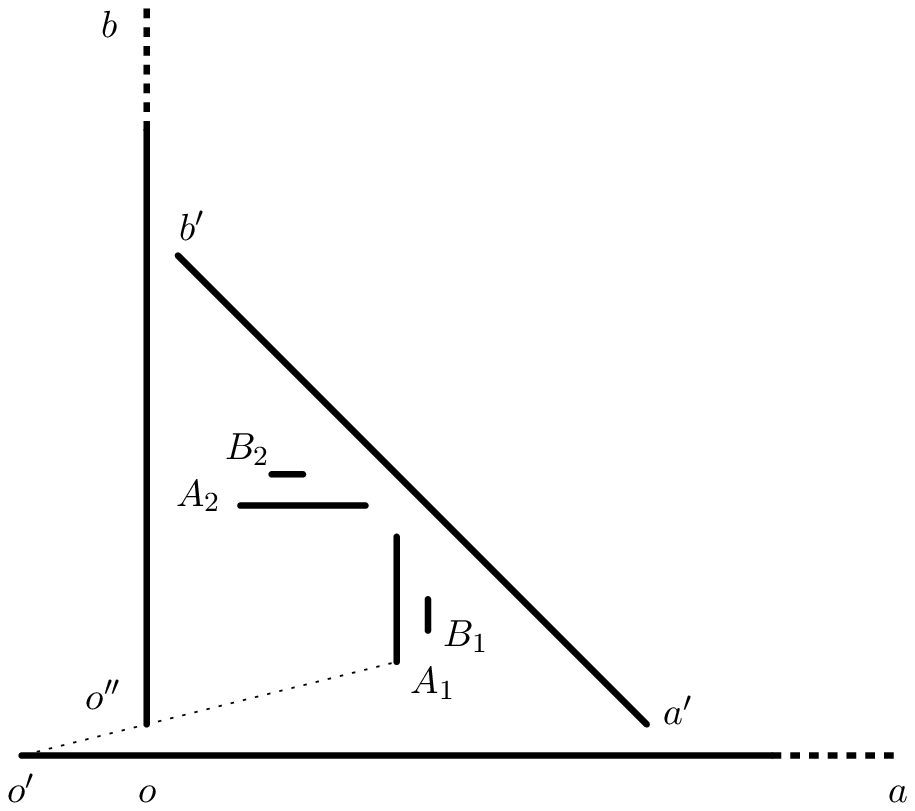}
	\hspace{\stretch1}
	\includegraphics[width=.48\textwidth]{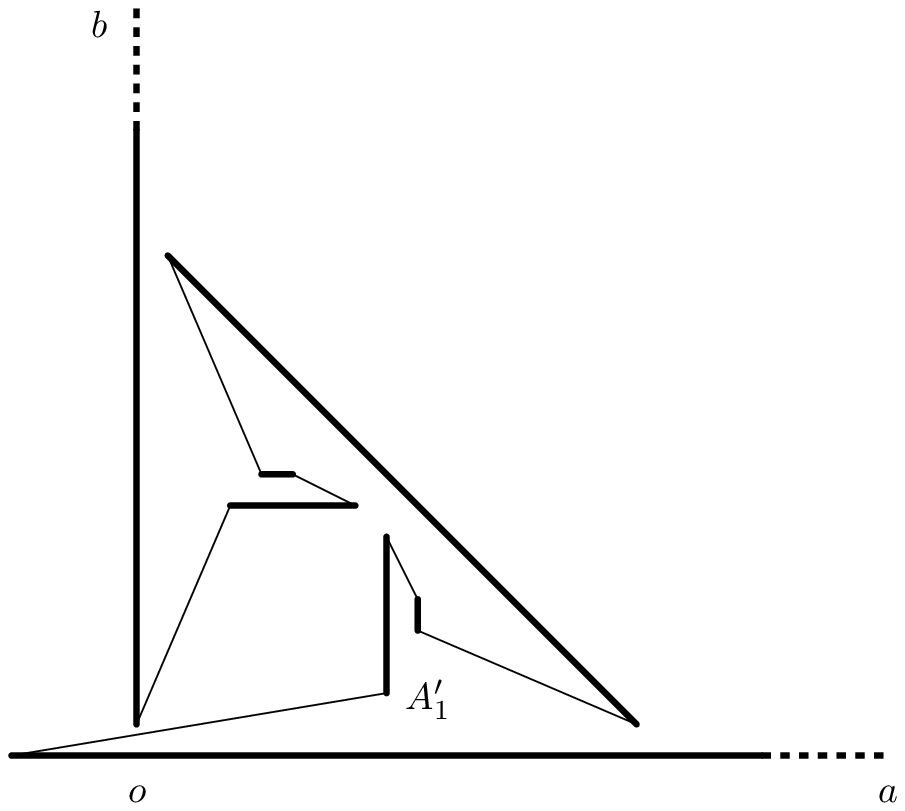}
	\caption{The gadget for the intersection $o$ between a horizontal segment $oa$
	and a vertical segment $ob$.
	Left:
	The coordinates of the segment endpoints relative to $o$ are
	$o = (0,0)$, $o' = (-4,0)$, $o'' = (0,1)$,
	$a' = (16,1)$, $b' = (1,16)$,
	$A_1 = 8 \times [3,7]$,
	$B_1 = 9 \times [4,5]$,
	$A_2 = [3,7] \times 8$,
	$B_2 = [4,5] \times 9$.
	The lower endpoint $a_1$ of $A_1$ is collinear with the two points $o'$ and $o''$.
	Right: An alternating path of segments and visibility edges in the gadget
	after extending $A_1$ to $A'_1$.
	The lower endpoint $a'_1$ of $A'_1$ can see the segment $o'a$
	at the endpoint $o'$.}\label{fig:gadget}
\end{figure}

By symmetry, it suffices to describe in detail only one variant of the gadget.
Refer to Figure~\ref{fig:gadget}.
For the intersection $o$ between a horizontal segment $oa$ and a vertical segment $ob$,
where $o$ is the left endpoint of $oa$ and the lower endpoint of $ob$,
the corresponding gadget is constructed as follows.
First,
separate the two segments $oa$ and $ob$ into two disjoint segments
$o'a$ and $o''b$,
by splitting their common endpoint $o$,
then moving one to $o'$ and the other to $o''$.
Next, add five disjoint segments
$A_1$, $B_1$, $A_2$, $B_2$, and $a'b'$.
Finally, extend the segment $A_1$ to a slightly longer segment $A'_1$
by moving its lower endpoint down for a distance of $\delta = \frac1{800 n^2}$,
from $a_1 = (8,3)$ to $a'_1 = (8,3-\delta)$.

In the presence of $\S'$,
we say that two points $p$ and $q$ can \emph{see} each other
if the open segment $pq$ is disjoint from all closed segments in $\S'$,
and we say that two segments $A$ and $B$ in $\S'$ can \emph{see} each other
if at least one of the four pairs of endpoints,
one of $A$ and one of $B$,
can see each other.
The gadgets we constructed have the following property of mutual invisibility:

\begin{lemma}\label{lem:gadget}
	The four segments $A'_1$, $B_1$, $A_2$, and $B_2$
	in each gadget cannot see any segments in other gadgets.
	Moreover, in each gadget,
	$B_1$ can only see $A'_1$ and $a'b'$,
	$B_2$ can only see $A_2$ and $a'b'$,
	$A_2$ cannot see $o'a$ and can see $o''b$ only at $o''$,
	and $A'_1$ can see $o'a$ only at $o'$.
\end{lemma}

\begin{figure}[htb]
	\centering\includegraphics{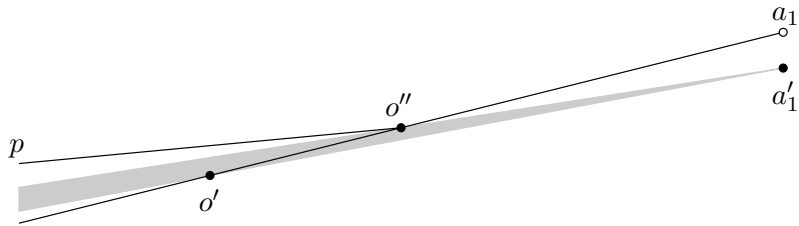}
	\caption{The narrow viewing angle from $a'_1$ between $o'$ and $o''$.}
	\label{fig:angle}
\end{figure}

\begin{proof}
	Clearly none of the eight endpoints of $A_1,B_1,A_2,B_2$
	in any gadget can see any endpoint outside the gadget.
	After extending $A_1$ to $A'_1$
	by moving its lower endpoint from $a_1$ to $a'_1$,
	the endpoint $a'_1$
	may see outside the gadget
	only through the gap between $o'$ and $o''$.
	Refer to Figure~\ref{fig:angle}.

	We next show that $a'_1$ still cannot see any integer endpoints from other gadgets.
	Note that
	$|a_1 a'_1| = \delta$,
	$|o'' a'_1| > 8$,
	and $\angle o'' a'_1 a_1 > \pi/2$.
	Thus
	$$
	\angle a_1 o'' a'_1 < \tan\angle a_1 o'' a'_1 < \frac{|a_1 a'_1|}{|o'' a'_1|} < \delta/8.
	$$

	Consider the ray that starts from $o''$ and goes through $o'$.
	Rotate this ray clockwise about $o''$ for a positive angle
	till it goes through another integer point $p$ in $V(\S')$.
	The area of the triangle $o'o''p$ is at least $1/2$
	since all three endpoints have integer coordinates.
	Recall that the endpoints of all segments in $\S$ are in the range
	$[1, 22n^2]\times [1, 11n]$,
	and is scaled by a factor of $40$
	by the transformation from $\S$ to $\S'$.
	Thus we have
	$|po''| < 40\cdot (22n^2 + 11n) < 40 \cdot 32 n^2$ for $n \ge 2$.
	Also note that $|o'o''| < 5$.
	Thus
	\begin{equation}\label{eq:scaling}
		\angle o'o''p > \sin\angle o'o''p
		= \frac{2\cdot\mathrm{area}(o'o''p)}{|o'o''|\cdot |po''|}
		> \frac1{5 \cdot 40 \cdot 32 n^2}
		= \frac1{8 \cdot 800 n^2}
		= \delta/8.
	\end{equation}

	Therefore we have $\angle a_1 o'' a'_1 < \angle o'o''p$.
	Since the three points $o'$, $o''$, and $a_1$ are collinear,
	this inequality implies that the line through $a'_1$ and $o''$
	splits the angle $\angle o'o''p$.
	Thus the cone with the angle $\angle o' a'_1 o''$ does not contain
	in its interior any integer endpoints in $V(\S')$.
	Thus $a'_1$ cannot see any integer endpoints from other gadgets.

	Moreover,
	no two endpoints $a'_1$ from two different gadgets can see each other,
	because the narrow ranges of slopes of viewing lines through $a'_1$,
	as illustrated by the dotted lines in Figure~\ref{fig:turn4},
	are disjoint for intersections of different orientations.

	The remaining claims of the lemma are easy to verify.
\end{proof}

The following lemma shows that our local transformation preserves linkability:

\begin{lemma}\label{lem:link}
	$\S$ can be linked into a simple polygon
	if and only if $\S'$ can be linked into a simple polygon.
\end{lemma}

\begin{proof}
	We first prove the direct implication.
	Suppose that $\S$ can be linked into a simple polygon.
	For each visibility edge of this polygon between two endpoints in $V(\S)$,
	we add a visibility edge between the corresponding endpoints in $V(\S')$.
	For each endpoint $o$ incident to two segments $oa$ and $ob$ in $\S$,
	we link the segments in the corresponding gadget in $\S'$
	following the alternating path as illustrated in Figure~\ref{fig:gadget} right.
	Then $\S'$ is also linked into a simple polygon.

	We next prove the reverse implication.
	Suppose that $\S'$ can be linked into a simple polygon.
	Recall Lemma~\ref{lem:gadget}.
	In any gadget,
	$B_1$ can only see $A'_1$ and $a'b'$,
	and $B_2$ can only see $A_2$ and $a'b'$.
	Thus $a'b'$ must be linked to both $B_1$ and $B_2$,
	which are linked to $A'_1$ and $A_2$, respectively.
	This chain can extend further only to $o'a$ and $o''b$, respectively,
	since $A_2$ cannot see $o'a$.
	Thus the seven segments in each gadget must be linked consecutively
	in the simple polygon along the sequence
	$o'a,A'_1,B_1,a'b',B_2,A_2,o''b$.
	Then, following the other visibility edges of the simple polygon through $\S'$,
	which are outside and between the gadgets,
	$\S$ can be linked into a simple polygon too.
\end{proof}

Recall that the coordinates of endpoints in $V(\S)$
are integers of magnitude $O(n^2)$.
After the transformation,
all endpoints in $V(\S')$ except $a'_1$
have integer coordinates too.
We can scale all coordinates by another factor of $1/\delta = 800 n^2$,
so that all endpoints in $V(\S')$ including $a'_1$
have integer coordinates of magnitude $O(n^4)$,
and the reduction remains strongly polynomial.
Thus \scirc\ is strongly NP-hard,
even if the input segments are disjoint and have only four distinct orientations.

\section{Modification for \spath}

To prove that \spath\ is also NP-hard,
we use almost the same transformation from $\S$ to $\S'$ as before,
except two changes:
\begin{enumerate}\setlength\itemsep{0pt}

		\item
			Increase the initial scaling factor from $40$ to $80$,
			and correspondingly decrease the distance $\delta$ between $a_1$ and $a'_1$
			from $\frac1{800 n^2}$ to $\frac1{1600 n^2}$
			when constructing the gadgets.

		\item
			Select an arbitrary gadget as illustrated in Figure~\ref{fig:gadget},
			and replace it by an extended gadget as illustrated in Figure~\ref{fig:gap}.

\end{enumerate}

The purpose of the first change is to make room for additional segments
in the extended gadget.
Updating the scaling factor and the distance $\delta$ together
ensures that \eqref{eq:scaling} still holds,
and hence Lemma~\ref{lem:gadget} remains valid.
Clearly, the reduction remains strongly polynomial.

\begin{figure}[htb]
	\includegraphics[width=.48\textwidth]{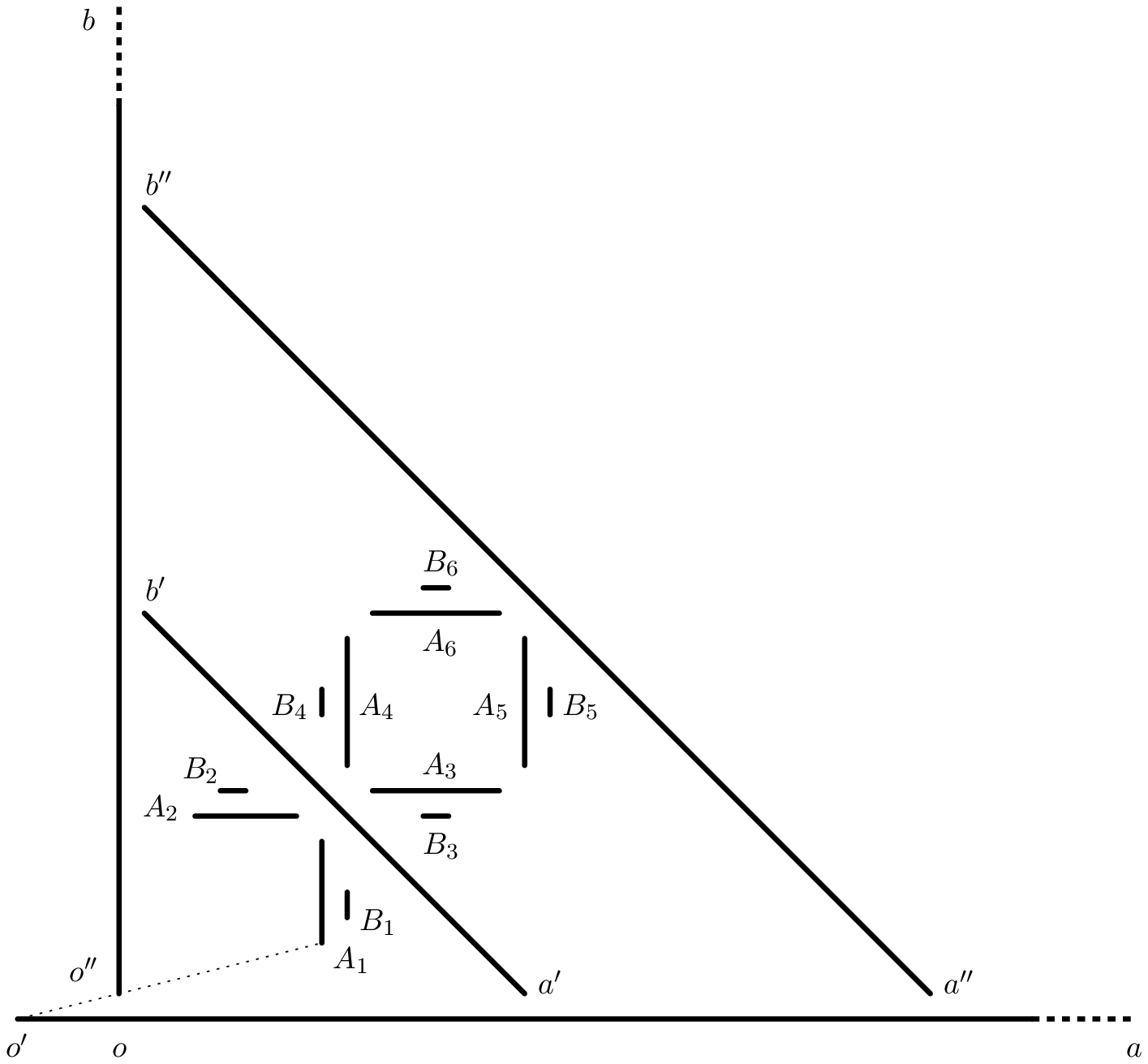}
	\hspace{\stretch1}
	\includegraphics[width=.48\textwidth]{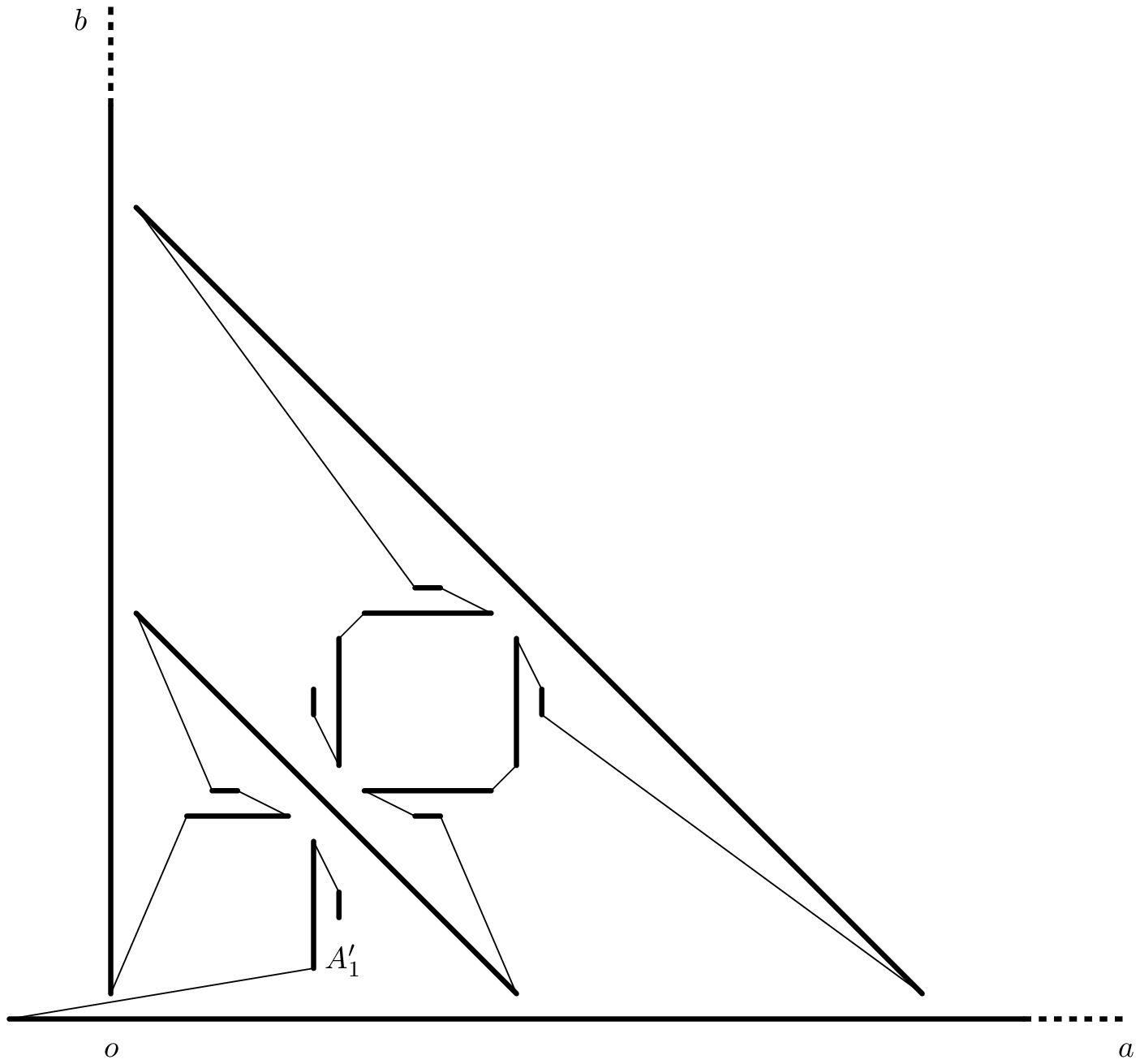}
	\caption{The extended gadget for the intersection $o$ between a horizontal segment $oa$
	and a vertical segment $ob$.
	Left:
	The coordinates of the points relative to $o$ are
	$o = (0,0)$, $o' = (-4,0)$, $o'' = (0,1)$,
	$a' = (16,1)$, $b' = (1,16)$,
	$A_1 = 8 \times [3,7]$,
	$B_1 = 9 \times [4,5]$,
	$A_2 = [3,7] \times 8$,
	$B_2 = [4,5] \times 9$,
	$a'' = (32,1)$, $b'' = (1,32)$,
	$A_3 = [10,15] \times 9$,
	$B_3 = [12,13] \times 8$,
	$A_4 = 9 \times [10,15]$,
	$B_4 = 8 \times [12,13]$,
	$A_5 = 16 \times [10,15]$,
	$B_5 = 17 \times [12,13]$,
	$A_6 = [10,15] \times 16$,
	$B_6 = [12,13] \times 17$.
	Right: An alternating path starting and ending in the extended gadget.}\label{fig:gap}
\end{figure}

We have the following lemma analogous to Lemma~\ref{lem:link}:

\begin{lemma}
	$\S$ can be linked into a simple polygon
	if and only if $\S'$ can be linked into a simple polygonal chain.
\end{lemma}

\begin{proof}
	We first prove the direct implication.
	Suppose that $\S$ can be linked into a simple polygon.
	We link the segments in $\S'$ as before, except that in the extended gadget
	we link the segments as illustrated in Figure~\ref{fig:gap} right.
	Then $\S'$ is linked into a simple polygonal chain starting and ending in the extended gadget.

	We next prove the reverse implication.
	Suppose that $\S'$ can be linked into a simple polygonal chain.
	In the extended gadget,
	$B_1$ can only see $A'_1$ and $a'b'$,
	$B_2$ can only see $A_2$ and $a'b'$,
	and moreover for $i \in \{3,4,5,6\}$,
	$B_i$ can only see $A_i$, $a'b'$, and $a''b''$.
	Since $a'b'$ and $a''b''$ together can be linked to at most four other segments
	in the polygonal chain,
	and since the polygonal chain has only two vertices of degree one,
	it follows that among the six segments $B_i$, $1 \le i \le 6$,
	each of $a'b'$ and $a''b''$ must be linked to two distinct segments,
	and the polygonal chain must start and end at the other
	two segments.
	Thus the $12$ segments $A'_1,B_1,A_2,B_2,A_3,B_3,A_4,B_4,A_5,B_5,A_6,B_6$,
	and the two segments $a'b',a''b''$,
	are all linked locally to each other in the extended gadget,
	and they are then linked to $o'a$ and $o''b$, respectively, at $o'$ and $o''$.
	As before,
	in each gadget other than this extended gadget,
	the seven segments must be linked consecutively
	along the sequence
	$o'a,A'_1,B_1,a'b',B_2,A_2,o''b$.
	Then, following the other visibility edges of the simple polygonal chain through $\S'$,
	which are outside and between the gadgets,
	$\S$ can be linked into a simple polygon.
\end{proof}

Thus \spath\ is also NP-hard.
This completes the proof of Theorem~\ref{thm:hard}.

\paragraph{Postscript}
In SoCG 2019, Akitaya et~al.~\cite{AKRST19,AKRST19b} presented a proof of the NP-hardness
of \scirc\ using a gadget very similar to ours.
We thank Adrian Dumitrescu for bringing it to our attention.

Compare~\cite[Figure~13]{AKRST19} and~\cite[Figure~19]{AKRST19b} with our Figure~\ref{fig:gadget}.
Note in particular that their $p_4 p_5$ and $p_{12} p_{13}$ correspond to
our $A_2$ and $A_1$,
their $p'_2$ and $p_2$ correspond to our $o''$ and $o'$,
and their properties (i) and (ii) correspond to our Lemma~\ref{lem:gadget}.
The main idea behind the internal connections of their gadget is the same as ours.
However, to ensure property (ii),
they use a more complicated method to custom-make each gadget
individually~\cite[Figure 20]{AKRST19b}.
Our construction handles all gadgets uniformly by the small distance $\delta$,
the narrow viewing angle in Figure~\ref{fig:angle},
and the disjoint ranges of slopes for the dotted lines in Figure~\ref{fig:turn4}.

To the best of our knowledge, our result on the NP-hardness of \spath\ is new.


\begin{thebibliography}{9}

	\bibitem{AKRST19}
		H. A. Akitaya, M. Korman, M. Rudoy, D. L. Souvaine, C. D. T\'oth.
		Circumscribing polygons and polygonizations for disjoint line segments.
		\emph{Proceedings of the 35th International Symposium on Computational Geometry},
		9:1--9:17, 2019.

	\bibitem{AKRST19b}
		H. A. Akitaya, M. Korman, M. Rudoy, D. L. Souvaine, C. D. T\'oth.
		Circumscribing polygons and polygonizations for disjoint line segments.
		arXiv:1903.07019, 2019.

	\bibitem{BHT01}
		P. Bose, M. E. Houle, and G. T. Toussaint.
		Every set of disjoint line segments admits a binary tree.
		\emph{Discrete and Computational Geometry},
		26:387--410, 2001.

	\bibitem{Ra89}
		D. Rappaport.
		Computing simple circuits from a set of line segments is NP-complete.
		\emph{SIAM Journal on Computing},
		18:1128--1139, 1989.

	\bibitem{RT86}
		P. Rosenstiehl and R. E. Tarjan.
		Rectilinear planar layouts and bipolar orientations of planar graphs.
		\emph{Discrete and Computational Geometry},
		1:343-353, 1986.

	\bibitem{To06}
		C. D. T\'oth.
		Alternating paths along axis-parallel segments.
		\emph{Graphs and Combinatorics},
		22:527--543, 2006.

\end{thebibliography}
\end{document}